\documentclass[letterpaper,10pt,conference]{ieeeconf}
\IEEEoverridecommandlockouts

\usepackage{amsmath,amsfonts,amssymb}
\usepackage{graphicx}

\newtheorem{theorem}{Theorem}
\newtheorem{lemma}{Lemma}

\newtheorem{proposition}{Proposition}
\newtheorem{corollary}{Corollary}

\newtheorem{assumption}{Assumption}
\newtheorem{remark}{Remark}

\usepackage[all=normal,paragraphs,floats,bibnotes,wordspacing,charwidths,leading,indent]{savetrees}

\usepackage{color}
\usepackage{comment}

\newcommand{\sam}[1]{\textcolor{black}{#1}}

\newcommand{\del}[1]{{\textcolor{red}
{\begin{tt}To modify: \end{tt}#1}}}

\def\R{\mathbb{R}}
\def\N{\mathbb{N}}

\def\G{\mathcal{G}}
\def\E{\mathcal{E}}
\def\C{\mathcal{C}}

\def\1{{\bf 1}}
\def\0{{\bf 0}}
\usepackage{accents}
\renewcommand{\underbar}[1]{{\underaccent{\bar}{#1}}}

\newcommand{\unm }{\{1,\ldots, m\}}
\newcommand{\unn}{\{1,\ldots, n\}}

\newcommand{\uc}{\underbar{c}}
\newcommand{\ua}{\underbar{a}}

\newcommand{\NN}{{\mathcal{N}}}
\newcommand{\MM}{{\mathcal{M}}}

\newcommand{\OO}{\mathcal{O}}

\newcommand{\GG}{\mathcal{G}}
\newcommand{\un}{{\bf{1}}}

\newcommand{\argmax}{\textrm{arg}\max}

\newcommand{\eps}{\varepsilon}

\newcommand{\ie}{{\it i.e., }}

\newcommand{\eg}{{\it e.g. }}

\newcommand{\w}[2]{w_{#1 \,\leftarrow\, #2}}

\newcommand{\ssum}{\displaystyle\sum}

\newcommand{\mmax}{\displaystyle\max}

\newcommand{\ubar}[1]{\underaccent{\bar}{#1}}

\begin{document}

\title{\LARGE \bf
Time scale modeling for consensus in sparse directed networks with time-varying topologies \thanks{The work of S. Martin and I.-C. Mor\u{a}rescu was supported by project PEPS MoMIS MADRES and PEPS INS2I CONAS funded by the CNRS and project \textit{Computation Aware Control Systems (COMPACS)}, ANR-13-BS03-004 funded by ANR. The work of D. Ne\v{s}i\'c was supported by the Australian Research Council under the Discovery Project DP1094326.}
}

\author{Samuel Martin$^*$, Irinel-Constantin Mor\u{a}rescu$^*$
\thanks{$^*$  Universit\'e de
Lorraine, CRAN, UMR 7039, and CNRS, CRAN, UMR 7039,  2, Avenue de la For\^et de Haye, 54500, Vand\oe uvre, France.
(\textit{samuel.martin},\textit{constantin.morarescu})@\textit{univ-lorraine.fr}.
}, Dragan Ne\v{s}i\'c$^{\#}$\thanks{$^{\#}$  Department of Electrical and Electronic Engineering, The University of Melbourne, Parkville, 3010 VIC, Australia. \textit{ dnesic@unimelb.edu.au}}
}

\maketitle
\thispagestyle{empty}
\pagestyle{empty}

\begin{abstract}
The paper considers the consensus problem in large networks represented by time-varying directed graphs. A practical way of dealing with large-scale networks is to reduce their dimension by collapsing the states of nodes belonging to densely and intensively connected clusters into aggregate variables. It will be shown that under suitable conditions, the states of the agents in each cluster  converge fast toward a local agreement. Local agreements correspond to aggregate variables which slowly converge to consensus. Existing results concerning the time-scale separation in large networks focus on fixed and undirected graphs. The aim of this work is to extend these results to the more general case of time-varying directed topologies. It is noteworthy that in the fixed and undirected graph case the average of the states in each cluster is time-invariant when neglecting the interactions between clusters. Therefore, they are good candidates for the aggregate variables. This is no longer possible here. Instead, we find suitable time-varying weights to compute the aggregate variables as time-invariant weighted averages of the states in each cluster. This allows to deal with the more challenging time-varying directed graph case. We end up with a singularly perturbed system which is analyzed by using the tools of two time-scales averaging which seem appropriate to this system.
\end{abstract}

\section{Introduction}

Large dynamical systems of interacting units such as power grids, transportation networks or the brain can be modeled as multi-agent systems. Unfortunately, these models consist of very large number of coupled differential equations whose analysis becomes intractable. Therefore, it is crucial to reduce the complexity of the system by proposing appropriate reduced-order models that still capture the  asymptotic behavior of the network. A natural approach is to merge a number of agents into a single node whose state approximates the states of the agents that generated it. This idea has been successfully applied for networks represented as fixed undirected graphs  \cite{Hanski1998,Gfeller2008,Arcak2007,ChowKokotovic,Bulo2012}. An extension of this idea to networks represented as fixed directed graphs was proposed in \cite{MP2012}.

The goal of this paper is to provide a methodology to generate appropriate reduced systems for multi-agent systems with a large number of agents interacting through directed time-varying networks. A fundamental assumption is that the interaction network consist of sparsely connected clusters of densely connected agents \cite{Arcak2007,ChowKokotovic,Morarescu-et-al-TAC2016}. An interesting feature of these networks is that the agents first converge fast toward a local agreement inside the clusters and then they slowly converge to a global consensus. This time scale separation has been emphasized and used for the analysis of power grids with fixed and undirected interconnections~\cite{Arcak2007,ChowKokotovic,Bulo2012}. The same feature has been used in \cite{MG10} for clusters detection in large scale networks. However, up to now, only symmetric and time-invariant communication networks have been dealt with.

In the case of fixed and undirected graph the average of the states in each cluster is time invariant when neglecting the interactions between clusters. Therefore, it is natural to use it as an aggregate state. When the interactions are directed and time varying we have to define appropriate aggregate variables. To do so, we introduce time-varying weights defining a time-invariant weighted average (using reasoning similar to \cite{MartinHendrickx2015}) as aggregate state of a cluster isolated from the other clusters. As a result we exhibit a time-scale separation in the network justifying the approximation of the overall network dynamics with the dynamics of the slow aggregate variables. It is important to note that unlike the fixed undirected interconnections graph~\cite{Arcak2007,ChowKokotovic,Bulo2012}, the resulting singularly perturbed system in this paper is not in standard form and we cannot use the standard tools \cite{KhalilBook,KokotovicBook}. Instead it is more natural to employ averaging theory for two time-scales systems~\cite{Balachandra,NesicAveraging2001,Teel}. \sam{Note that these results are not restricted to periodic systems.}
The contribution of the paper can be summarized as follows:
\begin{itemize}
\item we provide the assumptions guaranteeing that {\bf the convergence is faster inside clusters than between clusters};
\item under these assumptions, we show that {\bf time-scale separation occurs and can be used for analysis of directed and time varying interconnections;}

\item {\bf we use averaging theory for two time-scales systems} which seems to be appropriate, under the stated conditions, for the analysis of networks structured in sparsely connected clusters of densely connected agents.
\end{itemize}
The rest of the paper is organized as follows. Section \ref{Pb} introduces the model under study, the assumptions allowing to rewrite the model in singular perturbation form and presents the main contributions related to the singular perturbation modeling and analysis. In Section \ref{sec:model-reduction} we give some technical results concerning the two time-scales modeling. The singular perturbation analysis is provided in Section \ref{sec:main} and an illustrative numerical examples is presented in Section \ref{sec:numerical-illustration}.
\\
\textbf{Notation.}
The following standard notation will be used throughout the paper. The set of non-negative integers, real and non-negative real numbers are denoted by $\N,\ \R$ and $\R_+$, respectively. For a matrix $A$ we denote by $\|A\|_\infty$ its infinity norm. The transpose of a matrix $A$ is denoted by $A^\top$. By $I_k$ we denote the $k\times k$ identity matrix. $\1_k$ and $\0_k$ are the column vectors of size $k$ having all the components equal 1 and 0, respectively. A non trivial subset $S$ of a set $C$, denoted as $S \sqsubset C$, is a non-empty set with $S\subsetneq C$.
A \emph{directed path of length $p$} in a given directed graph $\G=(\NN,\E)$ is a union of directed edges $\bigcup_{k=1}^{p}(i_k,j_k)$ such that $i_{k+1}=j_k,\,\forall k\in\{1,\ldots,p-1\}$.
The node $j$ is  \emph{connected} to node $i$ in a directed graph $\G=(\NN,\E)$ if there exists at least a directed path in $\G$ from $i$ to $j$ (i.e. $i_1=i$ and $j_p=j$). 
\section{Preliminaries}\label{Pb}
\subsection{Model statement}

Let $\NN\triangleq\unn$ be a set of $n$ agents. By abuse of notation we denote  both the agent and its index by the same symbol $i \in \NN$. Each agent is characterized by a scalar state $x_i\in \R, \forall i\in\NN$ that evolves according to the following consensus system
\begin{equation}\label{eq:def_sys_derivative}
\dot{x}_i (t)= \sum_{j=1}^n a_{ij}(t) (x_j(t) - x_i(t)),\  \forall i\in\NN,
\end{equation}
where $a_{ij}(t)\geq 0$ are twice continuously differentiable and uniformly upper-bounded functions of time representing the \textit{interaction weights}. Let $x(t)=(x_1(t),\ldots,x_n(t))^\top\in\R^n$ be the overall state of the network collecting the states of all the agents. 
Under these conditions, there exists a unique differentiable function of time $x:\R^+ \rightarrow \R^n$ whose components satisfy equation~(\ref{eq:def_sys_derivative}) for all $t\in \R^+$~\cite{Filippov1988}, and it is bounded~\cite{HendrixTsitsiklis_CutBalanced_IEEETAC2013}.
We call it the {\it trajectory} of the overall system. We say the trajectory asymptotically reaches a consensus when there exists a common agreement value $\alpha\in\R$ such that
$$
 \lim_{t\rightarrow +\infty} x_i(t) = \alpha, \ \forall i \in \NN.
$$
In the sequel, agents are assumed to be partitioned in $m$ non-empty clusters: $\C_1, \C_2, \ldots, \C_m\subset\NN$, that are assumed to be given or can be easily identified. For instance clusters may correspond to groups of agents which are spatially close to each other while different clusters are distant to each other. Let us introduce the following supplementary notation: $\MM \triangleq\unm$ and $n_i$ denotes the cardinality of cluster $\C_i$ and $\underbar{n} = \min_{i \in \MM } n_i$. Without loss of generality, we permute the agents' labels according to the partition so that when $j\in \C_i$ and $j' \in \C_{i+1}$, $j<j'$.

Denote $A(t) = (a_{ij}(t))$ the adjacency matrix of communication weights at time $t$, $D(t)=diag(d_{ii}(t))$ with $d_{ii}(t)=\sum_{j\in\NN}a_{ij}(t)$, and $L(t)=D(t)-A(t)$ its associated Laplacian matrix. We denote as $A^I(t) = diag(A^1(t),\ldots,A^m(t))$ the block-diagonal intra-cluster adjacency matrix where for $k \in \MM $, $A^k(t) = (a_{ij}(t))_{(i,j) \in \C_k^2}$. We also denote by $L^k(t)$ the Laplacian matrix associated to $A^k(t)$ and $L^I(t) = diag(L^1(t),\ldots,L^m(t))$ the block-diagonal Laplacian matrix associated to $A^I(t)$. The inter-cluster communications are then described by $A^E(t) = A(t) - A^I(t)$ and its associated Laplacian matrix $L^E(t) = L(t) - L^I(t)$. The inter-cluster adjacency matrix has zero elements $A^E_{ij}(t)$ if $i$ and $j$ belong to the same cluster.
\subsection{Framework assumptions}
In the following let us introduce some notation and the main hypotheses of this work.
For two subsets of nodes $A,B \subset \NN$, the sum of communication weights from $B$ to $A$ is denoted as 
$$
\w{A}{B}(t) = \sum_{i \in A, j \in B} a_{ij}(t).
$$
Let us also recall some notation in \cite{Arcak2007,ChowKokotovic,Bulo2012} by adapting them to cope with non constant internal communications:
\begin{equation}\label{eq:delta-def}
\begin{split}
&c^I(t)=\min_{k\in\MM, S \sqsubset \C_k}\w{S}{(\C_k \setminus S)}(t),\\
&\gamma^E(t)=\max_{k\in\MM}\w{\C_k}{(\NN \setminus \C_k)}(t), \\ &\delta(t)=\frac{\gamma^E(t)}{c^I(t)}. 
\end{split}\end{equation}
\begin{assumption}[Intra-cluster communication]\label{ass:intra-reciprocity}
There exist constants $\uc\geq0, K_I\geq 1$ such that $c^I(t) \ge \uc, \forall t\geq 0$ and for any cluster $k \in \MM $ and for all non trivial subsets $S \sqsubset \C_k$,
\vspace{-0.2cm}
$$
\w{S}{(\C_k \setminus S)}(t) \le K_I \cdot \w{(\C_k \setminus S)}{S}(t),\ \forall t\ge 0.
$$
\end{assumption}
\vspace{0.3cm}
As proven in the sequel, Assumption \ref{ass:intra-reciprocity} ensures exponential consensus inside each cluster seen as an independent network (see Lemma \ref{l:convergence-exponential-rate}). This will justify the merging of all the nodes belonging to a cluster into one aggregate node.

\begin{assumption}[Inter-cluster communication]\label{ass:inter-reciprocity}
There exist constants $K_E\ge 1$ and $\ua\in (0,1)$ such that for all non trivial subset $S \sqsubset \MM $, for all $t\ge 0$, the two following equations hold :
\begin{equation}\label{eq:persistence-external-weights}
\sum_{k \in S} \w{\C_k}{(\NN \setminus \C_k)}(t) \le K_E \cdot \sum_{k \in S} \w{(\NN \setminus \C_k)}{\C_k}(t),
\end{equation}
\vspace{-0.2cm}
\begin{equation}\label{eq:reciprocity-external-weights}
\sum_{k\in S}\sum_{h\in\MM\setminus S}\w{\C_k}{\C_h}(t) \geq c^I(t) \eps \ua.
\end{equation}
\end{assumption}
\vspace{0.3cm}
Assumption~\ref{ass:inter-reciprocity} ensures consensus of the aggregate state associated with the clusters, \sam{where equation~\eqref{eq:persistence-external-weights} ensures reciprocity of the external interactions while equation~\eqref{eq:reciprocity-external-weights} guarantees their persistence when Assumption~\ref{ass:intra-reciprocity} is satisfied (see \eg\cite{SamAntoine_Persistent_SICON2013} for definitions).} Constant $\ua$ can be chosen arbitrarily small.
It is necessary that $K_I \ge 1$ and $K_E \ge 1$ and the equality corresponds to symmetric communication.
The next assumption ensures that the total communication weight which a cluster $\C_k$ receives cannot exceed a proportion of the weight received by any non-trivial subset of $\C_k$ from the rest of $\C_k$.

\begin{assumption}\label{ass:cI-bound-eps-uniform-bound-on-delta}
The ratio between internal and external interaction weights (defined in~\eqref{eq:delta-def}) is bounded by a constant bound $\eps>0$ : for all $t\ge 0$,
\vspace{-0.2cm}$$ \delta(t) \le \eps.$$
\end{assumption}

The purpose of Assumption~\ref{ass:cI-bound-eps-uniform-bound-on-delta} is to ensure that the partition in clusters corresponds to the distribution of communication weights. It prevents cases where two subsets of a cluster are more connected to the outside than with each other. As it will be detailed in section~\ref{sec:model-reduction}, Assumption~\ref{ass:cI-bound-eps-uniform-bound-on-delta} is necessary in order to exhibit a two time-scale separation of the system.

\begin{remark}\label{ass:rho}
Assumption~\ref{ass:cI-bound-eps-uniform-bound-on-delta} guarantees that for each cluster $k \in \MM $ and for all non trivial subsets $S \sqsubset \C_k$,
\vspace{-0.2cm}
\[
\w{\C_k}{(\NN \setminus \C_k)}(t) \le \eps \cdot \w{S}{(\C_k \setminus S)}(t),\ \forall t\ge 0.
\]
\end{remark}
\vspace{0.2cm}

\begin{remark}\label{Rem-ass3} Note that Assumption~\ref{ass:cI-bound-eps-uniform-bound-on-delta} does not follow from Assumptions~\ref{ass:intra-reciprocity} and~\ref{ass:inter-reciprocity} hold. For a counter-example see~\cite{MartinMorarescuNesic_technical_com}. 
\end{remark}

\begin{remark}
For instances of classes of communication weights satisfying Assumptions~\ref{ass:intra-reciprocity} and~\ref{ass:inter-reciprocity}, see~\cite{HendrixTsitsiklis_CutBalanced_IEEETAC2013}. Section~\ref{sec:numerical-illustration} also provides an instance where Assumptions~\ref{ass:intra-reciprocity},~\ref{ass:inter-reciprocity} and~\ref{ass:cI-bound-eps-uniform-bound-on-delta} are all satisfied.
\end{remark}

Using the formalism of communication in clusters, one has the following result.
\begin{proposition}\label{prop:consensus}
Supposing that Assumptions~\ref{ass:intra-reciprocity},~\ref{ass:inter-reciprocity} and~\ref{ass:cI-bound-eps-uniform-bound-on-delta} hold, the trajectory of system~\eqref{eq:def_sys_derivative} reaches consensus exponentially fast. 
\end{proposition}
\begin{proof}
As shown in~\cite[Theorem 1]{MartinMorarescuNesic_technical_com}, Assumptions~\ref{ass:intra-reciprocity},~\ref{ass:inter-reciprocity} and~\ref{ass:cI-bound-eps-uniform-bound-on-delta} ensure the cut-balance condition in \cite{HendrixTsitsiklis_CutBalanced_IEEETAC2013}. Moreover, \eqref{eq:reciprocity-external-weights} and the lower-bound imposed on $c^I(t)$ guarantee that there exists a strongly connected graph $\G=(\NN,\E)$ such that 
\[\int_0^ta_{ij}(s)ds\geq \eps t\ua\uc,\quad \forall (j,i)\in\E\] and following \cite[Proposition 4]{SamAntoine_Persistent_SICON2013} one obtains the exponential convergence toward consensus.
\end{proof}
\begin{remark} Proposition \ref{prop:consensus} applies in the more general setting of integral equations with non smooth $a_{ij}$ (see \cite{HendrixTsitsiklis_CutBalanced_IEEETAC2013}) but in the present study, we use only smooth $a_{ij}$ due to the singular perturbation analysis that is presented in the next section.\end{remark}

From Proposition~\ref{prop:consensus}, it is known that the trajectory of system~\eqref{eq:def_sys_derivative} will converge to consensus. It remains to characterize this convergence. In particular, we want to show that convergence occurs faster within clusters than between clusters. Therefore aggregating the nodes inside clusters yields a smaller dimension model that approximates the overall dynamics. The analysis is carried out for general time-varying, non-symmetric communication weights.
\subsection{Contributions}
Firstly, we show that under Assumptions~\ref{ass:intra-reciprocity},~\ref{ass:inter-reciprocity} and~\ref{ass:cI-bound-eps-uniform-bound-on-delta} the system  \eqref{eq:def_sys_derivative} rewrites in a singular perturbation form:
 \begin{equation}\label{eq:y-z-fast-dynamics-rescaled}
\left\{
\begin{split}
\frac{d\hat{y}}{dt_f}(t_f) &= \eps A_{11}(t_f,\eps)\hat{y}(t_f)+\eps A_{12}(t_f,\eps)\hat{z}(t_f),\\
\frac{d\hat{z}}{dt_f}(t_f)&= \eps A_{21}(t_f,\eps)\hat{y}(t_f)+A_{22}(t_f,\eps)\hat{z}(t_f),
\end{split}\right.
\end{equation} 
where $\eps$ is a small parameter, $t_f$ represents the fast time-scale that will be defined as well as all the matrices involved in the dynamics. The matrices $A_{ij}(t_f,\eps)$ will be shown to be continuously
differentiable in their arguments and the norms of matrices $A_{11}(t_f,\eps)$, $A_{12}(t_f,\eps)$ and $A_{21}(t_f,\eps)$ will be proven to be bounded above by a constant while the norm of $A_{22}(t_f,\eps)$ is bounded below away from $0$ (see equation~\eqref{eq:norm-Aij-bound}). The bounds are uniform in time $t_f$ and in $\eps$.
Roughly speaking, $\hat{y}\in\R^m$ and $\hat{z}\in\R^{n-m}$ describe the aggregate state of the clusters and the disagreement variables inside clusters, respectively. It is worth mentioning that we will emphasize the presence of a slow time-scale $t_s=\eps t_f$ and show that $\hat{y}$ evolves slowly while $\hat{z}$ evolves fast.
\begin{assumption}\label{ass:averaging}
Suppose that the following limit exists and is independent of $t_f$
\[
A_{av}=\lim_{T\rightarrow\infty}\frac{1}{T}\int_{t_f}^{t_f+T}A_{11}(s,0)ds.
\]
\end{assumption}
We define the reduced (slow) system and boundary layer (fast) system:
\begin{equation}\label{eq:y-z-fast-and-slow}
\left\{
\begin{split}
&\frac{dy_s}{dt_s}(t_s) = A_{av} y_s(t_s);\\
&\frac{dz_f}{dt_f}(t_f) = A_{22}\left(t_f,0\right) z_f(t_f);
\end{split}\right. 
\end{equation}
to approximate the solutions of system~\eqref{eq:y-z-fast-dynamics-rescaled}. The following approximation result is based on Theorem 1 in \cite{Balachandra}.
\begin{theorem}\label{th:main} 
Let $y_s(0) = y_{s,0} \in \R^m$ and $z_f(0) = z_{f,0} \in \R^{n-m}$ some fixed initial conditions to system~\eqref{eq:y-z-fast-and-slow}. Set the same initial conditions to system~\eqref{eq:y-z-fast-dynamics-rescaled}, independently of $\eps$. Then, under Assumptions~\ref{ass:intra-reciprocity},~\ref{ass:inter-reciprocity},~\ref{ass:cI-bound-eps-uniform-bound-on-delta} and~\ref{ass:averaging}, for any fixed $\eps > 0$,
system~\eqref{eq:y-z-fast-dynamics-rescaled} possesses unique bounded solutions $\hat{y}(t_f,\eps)$ and $\hat{z}(t_f,\eps)$ defined on $\R^+$ (here, the dependence in $\eps$ is made explicit), and system~\eqref{eq:y-z-fast-and-slow} possesses unique bounded solutions $y_s$ and $z_f$ defined on $\R^+$. These solutions satisfy
\begin{eqnarray*}
&& \lim_{\eps\rightarrow0} \sup_{t_f\in \R+} \|\hat{y}(t_f,\eps)-y_s(\eps t_f)\| = 0,\\ && \lim_{\eps\rightarrow0} \sup_{t_f\in \R+} \|\hat{z}(t_f,\eps)-z_f(t_f)\| = 0.
\end{eqnarray*}
Furthermore, for any fixed $\eps > 0$, the trajectories $\hat{y}(t_f,\eps)$ and $y_s(\eps t_f)$ converge exponentially fast to consensus and the trajectories $\hat{z}(t_f,\eps)$ and $z_f(t_f)$ converge exponentially fast to $0$, when $t_f \rightarrow \infty$. 
\end{theorem}

\section{Two time-scale modeling}\label{sec:model-reduction}

Using the matrix notation, system~\eqref{eq:def_sys_derivative} can be represented as
\vspace{-0.1cm}
\begin{equation}\label{eq:matrix-form}
\dot{x} = -L(t)x = -(L^I(t)+L^E(t))x.
\end{equation}

\subsection{Aggregate state}\label{sec:invariant-intra-cluster-dynamics}

When trying to reduce the model and carry out time-scale separation between internal and external dynamics, difficulties arrise when the internal communication are neither constant nor symmetric as it was assumed in~\cite{Arcak2007,ChowKokotovic,Bulo2012}. In the present section, we introduce some time-varying weights defining a time-invariant weighted average for stand alone clusters, which defines the aggregate state of the cluster.\\
Let us consider the isolated internal dynamics :
\begin{equation}\label{sys:isolated-internal-dynamics}
 \dot{\tilde{x}}_{\C_k}(t) = -L^k(t) \tilde{x}_{\C_k}(t),
\end{equation}
which corresponds to the internal dynamics $x_{\C_k}(t)$ when $L^E(t) = 0$, where $x_{\C_k}(t)$ denotes the vector collecting states $x_i(t)$ for $i\in \C_k$.
We define, for all $t,\tau \ge 0$ with $t\le \tau$, the fundamental matrix $\Phi^k(t,\tau)$ of system~\eqref{sys:isolated-internal-dynamics} such that
\begin{equation}\label{eq:fundamental-matrix}
 \tilde{x}_{\C_k}(\tau) = \Phi^k(t,\tau) \tilde{x}_{\C_k}(t).
\end{equation}
As a direct consequence of Lemma 6 in \cite{MartinHendrickx2015}, for all $t,\tau\ge 0$ with $\tau\ge t$, for any $i,j \in \C_k$, weight $\Phi^k_{ij}(t,\tau)$ is non-negative and
\begin{equation}\label{eq:Phi-sum-to-one}
\sum_{j \in \C_k} \Phi^k_{ij}(t,\tau) = 1. 
\end{equation}
\begin{lemma}\label{l:convergence-exponential-rate}
Under Assumption~\ref{ass:intra-reciprocity} (intra-cluster communication) system~\eqref{sys:isolated-internal-dynamics} converges to consensus at exponential rate.
\end{lemma}

Lemma~\ref{l:convergence-exponential-rate} is a direct consequence of Proposition~4~\cite{SamAntoine_Persistent_SICON2013}. In the sequel, we assume that Assumption~\ref{ass:inter-reciprocity} is satisfied so that $\tilde{x}_{\C_k}$ converges to a consensus at exponential speed.
We denote the limit matrix $\Phi^k(t,\infty) = \lim_{\tau \rightarrow \infty} \Phi^k(t,\tau)$, which exists thanks to Lemma~\ref{l:convergence-exponential-rate}. Moreover, since $\tilde{x}_{\C_k}$ converges to a consensus independently of the inital conditions, there exists $q^k(t) \in \R^{n_k}$ such that
\vspace{-0.2cm}
\begin{equation}\label{eq:def_of_q_k}
\Phi^k(t,\infty) = \un \cdot (q^k(t))^T.
\end{equation}
Vector $q^k(t)$ plays an important role in the rest of the study, in particular it serves to define an invariant for the isolated internal dynamics~\eqref{sys:isolated-internal-dynamics} as given in the next lemma.

\begin{lemma}\label{l:invariant_qk_and_bounds}
The vector $q^k(t)$ satisfies $$q^k(t)^T \cdot \un = 1$$ and its components $q_i^k(t)$ are positive uniformly bounded \ie for all $t\ge 0$ and $i \in \C_k$,
$$ q_i^k(t) \in [q_{min},q_{max}] \in (0,1),$$
\vspace{-0.3cm}
with
\begin{equation*}
\begin{array}{l}
q_{min} =\left(\exp(-K_I) / \bar{n}\right)^{\bar{n}-1},
\; q_{max} = 1 - (\underbar{n} - 1) q_{min},
\end{array}
\end{equation*}
where $\underbar{n} = \min_{k \in \{1,\ldots,m\}} n_k$ and $\bar{n} = \max_{k \in  \MM } n_k$.
Furthermore, 
quantity $q^k(t)^T \cdot \tilde{x}_{\C_k}(t)$ is invariant in time, \ie
$$
\forall t\ge 0, q^k(t)^T \cdot \tilde{x}_{\C_k}(t) = q^k(0)^T \cdot \tilde{x}_{\C_k}(0).
$$
\end{lemma}
\vspace{0.2cm}
\begin{proof}
See Appendix.
\end{proof}
\begin{remark}\label{r:computing-qk}
Computing vector $q^k(t)$ is straightforward when the internal communication weights are constant up to a multiplicative factor, \ie $L^k(t) = l^k(t) L^k(0)$, with $l^k(t) \in \R$ (but not necessarilly symmetric). In this case, $q^k$ is constant, obtained as the left eigenvector of $L^k(0)$ associated to the eigenvalue $0$ such that $(q^k)^T \un_{n_k} = 1$ (see for instance \cite{Olfati-Saber2004}). More generally, $q^k(t)$ can be computed whenever the internal dynamics is simple enough to be computed analytically. In other cases, the time-scale separation still occurs naturally as the present study shows and Lemma~\ref{l:invariant_qk_and_bounds} provides bounds for the invariant $q^k(t)^T \cdot \tilde{x}_{\C_k}(t)$.
\end{remark}

In the following we introduce variables $y(t)\in\R^m$ which corresponds to the cluster-aggregate state collecting weighted averages of the agents' states in each cluster:
\[y_k(t)=(q^k(t))^\top x_{\C_k}(t)\quad \mbox{or} \quad y(t)=J(t)x(t)\in\R^m, \] with $x_{\C_k}$ the vector collecting the states $x_i$ of the nodes $i\in\C_k$ and where  $J(t)$ is the following block-diagonal matrix:
\begin{equation}\label{JH matrix}
J(t)=diag\big((q^1(t))^\top,\
(q^2(t))^\top,\ldots,
(q^m(t))^\top\big)\in\R^{m\times n}.
\end{equation}
The next lemma states that in each cluster $k$, the weight vector $q^k(t)$ is taken so that $y$ is invariant through system~\eqref{eq:matrix-form} when no inter-cluster communication takes place. This justifies why $y(t)$ is a good candidate to represent the cluster-aggregate states.

\begin{lemma}\label{l:LE0-implies-doty-0}
Suppose Assumption~\ref{ass:cI-bound-eps-uniform-bound-on-delta} holds. When $\eps$ (defined in Assumption~\ref{ass:cI-bound-eps-uniform-bound-on-delta}) equals $0$, no inter-cluster communication takes place, and in this case, $y$ is invariant in time \ie 
\vspace{-0.2cm}
\begin{equation}\label{eq:LE0-implies-doty-0}
\eps = 0 \Rightarrow L^E(t)=0 \Rightarrow \dot{y}(t)=0.
\end{equation}
In this case, if Assumptions~\ref{ass:intra-reciprocity} (intra-cluster communication) is also satisfied, the trajectory $x_{\C_k}$ in each cluster $k$ is identical to $\tilde{x}_{\C_k}$ and converges at exponential speed to a local consensus with consensus value $y_k(0)$.
\end{lemma}

\begin{proof}
The first part of statement~\eqref{eq:LE0-implies-doty-0} is a consequence of the second inequality in Assumption~\ref{ass:inter-reciprocity}.
The second part of statement~\eqref{eq:LE0-implies-doty-0} is a direct consequence of Lemma~\ref{l:invariant_qk_and_bounds}. The convergence to consensus comes from Lemma~\ref{l:convergence-exponential-rate}. 
Denote $x_{\C_k}^* \in \R$ the consensus value. Then, $\lim_{t\rightarrow \infty} x_{\C_k}(t) = x_{\C_k}^* \un$ and $$y_k(0) = \lim_{t\rightarrow \infty} y_k(t) = x_{\C_k}^* \lim q^k(t)^T \un_{n_k} = x_{\C_k}^*,$$
where we successively used equation~\eqref{eq:LE0-implies-doty-0}, the definition of $y_k$ and Lemma~\ref{l:invariant_qk_and_bounds}.
\end{proof}

\subsection{Disagreement variables}\label{sec:time-scale-separation}

In the following, we provide a change of variable adapted to the clustered communication with non-symmetric and time-varying weights. First, 
In order to characterize the disagreement inside each cluster, let us introduce the variable $z(t) \in \R^{n-m}$ which represents for all $k\in\MM$ the distances to the first agent in $\C_k$ :
$$z(t) = Qx(t),$$
with $Q = diag(Q_1,Q_2, \ldots, Q_m)$ where
\[
Q_k=\left[
\begin{array}{ccccc}
-1& 1 & 0 &\ldots & 0\\
 -1& 0 & 1 &\ldots &0\\
 \vdots & \vdots & \vdots &\vdots& \vdots\\
 -1& 0 & 0 &\ldots &1\\
\end{array}\right]\in\R^{(n_k-1)\times n_k}.\ 
\]
The initial state variable $x(t)$ can be written in function of $y(t)$ and $z(t)$ as follows :
\begin{equation}\label{eq:x-function-of-y-and-z}
x(t)=Hy(t)+\tilde{Q}(t) z(t),
\end{equation}
where $H=diag\big(\1_{n_1},\1_{n_2},\ldots,\1_{n_m}\big)\in\R^{n\times m}$, $\tilde{Q}(t)=diag(\tilde{Q}_1(t),\tilde{Q}_2(t),\ldots,\tilde{Q}_m(t))$ and
$\tilde{Q}_k(t) \in\R^{n_k\times (n_k-1)}$ such that
\[
\tilde{Q}_k(t)=\left[\begin{array}{cccc}
-q^k_2(t)& -q^k_3(t)  &\ldots & -q^k_{n_k}(t)\\
1-q^k_2(t)& -q^k_3(t)  &\ldots &-q^k_{n_k}(t)\\
-q^k_2(t)& 1-q^k_3(t)  &\ldots &-q^k_{n_k}(t)\\
 \vdots & \vdots  &\vdots& \vdots\\
-q^k_2(t)& -q^k_3(t)  &\ldots &1-q^k_{n_k}(t)\\
\end{array}\right].
\]
Matrices $Q,\tilde{Q}(t)$ were chosen to verify the following lemma.
\begin{lemma}\label{Q-properties}
The following properties hold true
\begin{itemize}
\item $Q_k\1_{n_k}=0,\ (q^k(t))^\top\tilde{Q}_k(t)=0,$
\item $Q_k\tilde{Q}_k(t)=I_{n_k-1}$
\item $\tilde{Q}_k(t) Q_k=I_{n_k}-\1_{n_k}(q^k(t))^\top$
\end{itemize}
\end{lemma}
\begin{proof}
The first property is obvious. The second comes from property $(q^k(t))^\top\1_{n_k}=1$ meaning that $\sum_{i=1}^{n_k}q_i^k(t)=1$ (see Lemma~\ref{l:invariant_qk_and_bounds}). For the third one, we compute each element of the first column of the product $\tilde{Q}_k(t) Q_k$. The first element is $\sum_{i=2}^{n_k}q_i^k(t)$ and is replaced by $1-q_1^k(t)$ while the rest of them are $\sum_{i=2}^{n_k}q_i^k(t)-1$ which is replaced by $-q_1^k(t)$.
\end{proof}

\begin{remark}
We note that definitions of $y(t)$ and $z(t)$ coincide with those used in \cite{ChowKokotovic} if we consider that the interactions are symmetric and constant i.e. $a_{ij}(t)=a_{ji}(t)$.  We also note that, for constant symmetric interactions, the variables $y(t)$ and $z(t)$ can be introduced using another metric 
 \cite{Arcak2007} such that $\tilde{Q}=Q^\top$. 
 \end{remark}

Using equation~\eqref{eq:x-function-of-y-and-z}, system~\eqref{eq:matrix-form} can be rewritten to make explicit the dynamics in terms of variables $y(t) = J(t) x(t)$ and $z(t) = Q x(t)$ :
\vspace{-0.2cm}
\begin{equation}\label{eq:y-z-dynamics}
\left\{
\begin{split}
\dot{y}(t)&= \bar{A}_{11}(t)y(t)+\bar{A}_{12}(t)z(t),\\
\dot{z}(t)&= \bar{A}_{21}(t)y(t)+\bar{A}_{22}(t)z(t),
\end{split}\right.
\end{equation}
\vspace{-0.1cm}
where
\vspace{-0.0cm}
\begin{equation*}\left\{
\begin{split}
\bar{A}_{11}(t) = -J(t)L(t)H, \quad \bar{A}_{12}(t) = -J(t)L(t)\tilde{Q}(t), \\
\bar{A}_{21}(t) = -QL(t)H, \quad \bar{A}_{22}(t) = -QL(t)\tilde{Q}(t).
\end{split}\right.
\end{equation*}

\subsection{Time-scale separation}

The next lemma shows that the aggregate variable $y$ and the intra-cluster disagreement variable $z$ evolve according to different time-scale, parametrized by factor $\eps$.
\begin{lemma}\label{l:orders}
The $\infty$-norms of the matrices in \eqref{eq:y-z-dynamics} satisfy
\begin{equation}
\begin{array}{lllll}
\hspace{-2mm}\|\bar{A}_{11}(t)\|_\infty&\hspace{-2mm}=&\hspace{-2mm}\|J(t)L(t)H\|_\infty&\hspace{-2mm}\leq&\hspace{-2mm}2c^I(t)\eps, \\
\hspace{-2mm}\|\bar{A}_{12}(t)\|_\infty&\hspace{-2mm}=&\hspace{-2mm}\|J(t)L(t)\tilde{Q}(t)\|_\infty&\hspace{-2mm}\leq&\hspace{-2mm}2c^I(t)\eps,\\
\hspace{-2mm}\|\bar{A}_{21}(t)\|_\infty&\hspace{-2mm}=&\hspace{-2mm}\|QL(t)H\|_\infty&\hspace{-2mm}\leq&\hspace{-2mm}2c^I(t)\eps, \\
\hspace{-2mm}\|\bar{A}_{22}(t)\|_\infty&\hspace{-2mm}=&\hspace{-2mm}\|QL(t)\tilde{Q}(t)\|_\infty&\hspace{-2mm}\geq&\hspace{-2mm}(1-8\eps) c^I(t). \\
\end{array}
\end{equation}
\end{lemma}
\vspace{0.2cm}
The proof of the lemma is given in the Appendix.

The only matrix in system~\eqref{eq:y-z-dynamics} which is not $\OO(c^I(t) \eps)$ is $\bar{A}_{22}$ which corresponds to the influence of $z$ on its dynamics. As a consequence, Lemma~\ref{l:orders} shows that variables $y$ behave as slow variables compared to the $z$ variables. To reveal this fact, we rescale the time with a fast time scale $t_f = \int_0^t c^I(s)ds$ and a slow time scale $t_s = \eps t_f$. 
This rescaling is well defined since according to Assumption~\ref{ass:intra-reciprocity}, $t_f$ diverges when $t\rightarrow \infty$. The transformation $t \mapsto t_f$ is invertible and so is $t \mapsto t_s$. We denote $t=\psi(t_f)$ the reverse transformation. 
Then
$y(t)=y(\psi(t_f))\triangleq\hat{y}(t_f), \ z(t)=z(\psi(t_f))\triangleq\hat{z}(t_f)$. Matrices $\bar{A}_{ij}$ are also rescaled as
\begin{equation*}
\left\{
\begin{array}{cc}
A_{11}(t_f,\eps) = \frac{\bar{A}_{11}(\psi(t_f))}{c^I(\psi(t_f)) \cdot \eps}, & A_{12}(t_f,\eps) = \frac{\bar{A}_{12}(\psi(t_f))}{c^I(\psi(t_f)) \cdot \eps}, \\
A_{21}(t_f,\eps) = \frac{\bar{A}_{21}(\psi(t_f))}{c^I(\psi(t_f)) \cdot \eps}, & A_{22}(t_f,\eps) = \frac{\bar{A}_{22}(\psi(t_f))}{c^I(\psi(t_f))}.
\end{array}
\right.
\end{equation*}
Thanks to Assumption~\ref{ass:cI-bound-eps-uniform-bound-on-delta}, the rescaling is well defined. At this point, the dependance in $\eps$ has been made explicit to stress that the singular perturbation approximation corresponds to set $\eps = 0$ (see below equation~\eqref{eq:dot-zf}).

Without entering in details, writing the  system~\eqref{eq:y-z-dynamics} in the slow time-scale leads to non-standard singularly perturbed system
\sam{because of the appearance of the term $1/\eps$ inside the matrix arguments. Thus singular perturbation analysis in~\cite{KhalilBook} cannot be applied}.
Nevertheless, using averaging theory \cite{Balachandra,Teel} for singularly perturbed systems we can approximate the behavior of system \eqref{eq:y-z-dynamics} using the fast-time scale dynamics.
\begin{lemma}\label{l:y-z-dynamics-in-fast-time-scale}
In the fast time-scale system~\eqref{eq:y-z-dynamics} is described by
\begin{equation}\label{eq:y-z-fast-dynamics-rescaled-2}
\left\{
\begin{split}
\frac{d\hat{y}}{dt_f}(t_f) &= \eps A_{11}(t_f,\eps)\hat{y}(t_f)+\eps A_{12}(t_f,\eps)\hat{z}(t_f),\\
\frac{d\hat{z}}{dt_f}(t_f)&= \eps A_{21}(t_f,\eps)\hat{y}(t_f)+A_{22}(t_f,\eps)\hat{z}(t_f),
\end{split}\right.
\end{equation} 
where the $\infty$-norms of the matrices in \eqref{eq:y-z-fast-dynamics-rescaled} satisfy
\begin{equation}\label{eq:norm-Aij-bound}
\begin{array}{lll}
\|A_{11}(t_f,\eps)\|_\infty\leq2,\quad\;
\|A_{12}(t_f,\eps)\|_\infty\leq2, \\
\|A_{21}(t_f,\eps)\|_\infty\leq2,\quad\;
\|A_{22}(t_f,\eps)\|_\infty\geq(1-8\eps).
\end{array}
\end{equation}
\end{lemma}
\vspace{0.2cm}
\begin{proof}
Using equation~\eqref{eq:y-z-dynamics}, 
\begin{eqnarray*}
\frac{d\hat{y}}{dt_f}(t_f) &=& \frac{d(y \circ \psi)}{dt_f}(t_f)
= \frac{d\psi}{dt_f}(t_f) \frac{dy}{dt}(\psi(t_f)) \\
&=& \frac{ (\bar{A}_{11}(\psi(t_f))y(\psi(t_f))+\bar{A}_{12}(\psi(t_f))z(\psi(t_f)))}{c^I(\psi(t_f))}.
\end{eqnarray*}
Moreover,
$$\bar{A}_{11}(\psi(t_f)) / c^I(\psi(t_f)) = \eps A_{11}(t_f,\eps),$$
and analogous equalities hold for $\bar{A}_{12}$ and $\bar{A}_{22}$ while
$$\bar{A}_{22}(\psi(t_f)) / c^I(\psi(t_f)) = A_{22}(t_f,\eps).$$
The bounds on the $\infty$-norms directly come from Lemma~\ref{l:orders}.
\end{proof}

\section{Singular perturbation analysis}\label{sec:main}

The dynamics given in Lemma~\ref{l:y-z-dynamics-in-fast-time-scale} fits in the generic form of systems treated in \cite{Balachandra,Teel}.
Let $z_f$ be the solution of 
\begin{equation}\label{eq:dot-zf}
\frac{dz_f}{dt_f}(t_f) = A_{22}\left(t_f,0\right) z_f(t_f).
\end{equation}
Note that for $\eps=0$ the matrix $A_{22}$ and corresponds to intra clusters communication only and $z_f(t_f) = Q\tilde{x}(\psi(t_f))$. So, as given in Lemma~\ref{l:LE0-implies-doty-0}, system~\eqref{eq:dot-zf} converges to $0$ exponentially fast. 
Denote $A_{11}(s,0) = \lim_{\eps \rightarrow 0} A_{11}(s,\eps)$, which exists thanks to equation~\ref{eq:norm-Aij-bound}.

Let $A_{av}$ be defined in Assumption \ref{ass:averaging} and $y_s$ be the solution of 
\begin{equation}\label{eq:dot-ys}
\frac{dy_s}{dt_s}(t_s) = A_{av} y_s(t_s).
\end{equation}
\begin{proposition}\label{prop:consensus-ys}
Denote $y_{s,k}$ for $k\in \MM$ the $k$-th component of $y_s$. Under Assumption~\ref{ass:averaging},
system~\eqref{eq:dot-ys} is a consensus system described element-wise by
\begin{equation}\label{eq:dot-ys-component-wise}
\frac{dy_{s,k}}{dt_s}(t_s) = \ssum_{j\in \MM} a^s_{kl} (y_{s,l}(t_s)-y_{s,k}(t_s)), k\in \MM,
\end{equation}
where $a^s_{kl} = \lim_{T\rightarrow\infty}\frac{1}{T}\int_{t_f}^{t_f+T} \ssum_{j \in \C_l}\ssum_{i \in \C_k}  \frac{q_i^k(\psi(s)) a_{ij}(\psi(s))}{c^I(\psi(s)) \eps}$ is independent of $t_f$.
Moreover, under Assumption \ref{ass:inter-reciprocity} (Inter-cluster communication) $y_s$ reaches consensus exponentially fast.
\end{proposition}
\begin{proof}
To show that system \eqref{eq:dot-ys} is a consensus system, it suffices to show that system $\dot{y} = A_{11}(t_f,\eps) y$ is a consensus system. Recall that $A_{11}(t_f,\eps) = \frac{\bar{A}_{11}(t)}{c^I(t) \eps}$ where recall that $t = \psi(t_f)$. Let $y \in \R^m$. Since $-\bar{A}_{11}(t) y = J(t)L(t)H y$, we first compute $(L(t)H y)_i$ for $i \in \R^n$. Let $k \in \MM $ such that $i \in \C_k$, we obtain
\begin{eqnarray*}
(L(t)H y)_i = \ssum_{l=1}^m \ssum_{j \in \C_l} L_{ij}(t) (Hy)_j,
\end{eqnarray*}
and for $j \in \C_l$,
$
(Hy)_j = \sum_{f=1}^m H_{jf}(t) y_f = y_l,
$
so that,
\begin{eqnarray*}
&&(L(t)H y)_i = \ssum_{l=1}^m \ssum_{j \in \C_l} L_{ij}(t) y_l\\[-2mm]
&&= \big( \ssum_{j \in \C_k} L_{ij}(t) \big) y_k +
\ssum_{\substack{l=1\\l\neq k}}^m \big( \ssum_{j \in \C_l} L_{ij}(t) \big) y_l\\[-2mm]
&&= - \ssum_{\substack{j \in \C_k \\ j \neq i} } a_{ij}(t) y_k
+ L_{ii}(t) y_k   
-\ssum_{\substack{l=1\\l\neq k}}^m  \ssum_{j \in \C_l} a_{ij}(t)  y_l,
\end{eqnarray*}
and since $L_{ii}(t) = \sum_{j \in \NN} a_{ij}(t)$, the previous equation simplifies to
\begin{equation*}
(L(t)H y)_i = \ssum_{\substack{l=1\\l\neq k}}^m \big( \ssum_{j \in \C_l} a_{ij}(t) \big) \left(y_k - y_l\right). 
\end{equation*}
Coming back to the computation of $-\bar{A}_{11}(t) y$, for $k \in \MM $, we have
\begin{eqnarray*}
(-\bar{A}_{11}(t) y)_k \hspace{-0.3cm}&=\hspace{-0.3cm}& \hspace{-0.3cm}\ssum_{i \in \NN} J_{ki}(t) (L(t)Hy)i= \ssum_{i \in \C_k} q_i^k(t) (L(t)Hy)i \\ [-2mm] 
&=& \ssum_{i \in \C_k} q_i^k(t) \ssum_{\substack{l=1\\l\neq k}}^m \big( \ssum_{j \in \C_l} a_{ij}(t) \big) \left(y_k - y_l\right) \\[-2mm] 
&=& \ssum_{\substack{l=1\\l\neq k}}^m \big( \ssum_{i \in \C_k} q_i^k(t) \ssum_{j \in \C_l} a_{ij}(t) \big) \left(y_k - y_l\right) \\[-3mm]
&=& c^I(t) \eps\ssum_{\substack{l=1\\l\neq k}}^m  a^s_{kl}(t_f) \left(y_k - y_l\right).
\end{eqnarray*}
Dividing by $c^I(t) \eps$ provides the first part of the result. To obtain the exponential convergence to consensus, we use Assumption~\ref{ass:inter-reciprocity} and the fact that $q_i^k(t) \in [q_{min},q_{max}]$ (see Lemma~\ref{l:invariant_qk_and_bounds}) to show that weights $a^s_{kl}$ satisfy an assumption analogous to Assumption~\ref{ass:intra-reciprocity} so that Proposition 4 in~\cite{SamAntoine_Persistent_SICON2013} applies to system~\eqref{eq:dot-ys}, as in Lemma~\ref{l:convergence-exponential-rate}. In particular, equations~\eqref{eq:persistence-external-weights} and~\eqref{eq:reciprocity-external-weights} allow showing 
that the persistent graph corresponding to weights $a^s_{kl}$ is strongly connected and that these weights satisfy the cut-balance property, respectively.
\end{proof}

We are now ready to prove the main result of the study.
{\it Proof of Theorem \ref{th:main}:}
We want to apply Theorem 1 in \cite{Balachandra}. Notice that system~\eqref{eq:y-z-fast-dynamics-rescaled} rewrites as
\vspace{-0.2cm}
\begin{equation}\label{eq:y-z-fast-dynamics-rescaled_Bala}
\left\{
\begin{split}
\frac{d\hat{y}}{dt_f}(t_f) &= \eps f(t_f,t_s,\hat{y},\hat{z},\eps),\\
\frac{d\hat{z}}{dt_f}(t_f)&= g(t_f,t_s,\hat{y},\hat{z},\eps),
\end{split}\right.
\end{equation} 
where, following the notation in \cite{Balachandra},
\begin{equation}
 \begin{array}{l}
f(t_f,t_s,\hat{y},\hat{z},\eps) = A_{11}(t_f,\eps)\hat{y}(t_f)+ A_{12}(t_f,\eps)\hat{z}(t_f),\\
g(t_f,t_s,\hat{y},\hat{z},\eps) = \eps A_{21}(t_f,\eps)\hat{y}(t_f)+A_{22}(t_f,\eps)\hat{z}(t_f)
 \end{array}
\end{equation}
It is worth noting that in our case $f$ and $g$ do not depend explicitly on the slow time $t_s$.
Equation $\dot{z}_f = g(t_f,\tau,y,z_f,0)$ where $\tau$ and $y$ are regarded as parameters corresponds to equation~\eqref{eq:dot-zf}. According to Lemma~\ref{l:LE0-implies-doty-0}, $z_f$ converges at exponential speed to $0$ so that Hypothesis H1 in~\cite{Balachandra} is satisfied. Moreover, $f_0$ in H2 in~\cite{Balachandra} is
\begin{eqnarray*}
f_0(\tau,y) &=& \lim_{T \rightarrow \infty} \frac{1}{T} \int_{t}^{t+T} f(s,\tau,y,z_f(s),0)ds\\
&=& \lim_{T \rightarrow \infty} \frac{1}{T} \int_{t}^{t+T} A_{11}(s,0)y+ A_{12}(s,0)z_f(s) ds\\
&=& \lim_{T \rightarrow \infty} \frac{1}{T} \int_{t}^{t+T} A_{11}(s,0)ds \cdot y
= A_{av} \cdot y
\end{eqnarray*}
where we used the exponential decrease of $z_f$ and the bound on $\|A_{12}(s,0)\|_{\infty}$ in equation~\eqref{eq:norm-Aij-bound} to remove the second term under the integral. Assumption~\ref{ass:averaging} shows that H2 and equation~(2.4) in~\cite{Balachandra} are satisfied. The averaged equation (2.5) in~\cite{Balachandra} corresponds to equation~\eqref{eq:dot-ys}, whose trajectory are bounded and converge to consensus at exponential speed thanks to Proposition~\ref{prop:consensus-ys}. As a consequence, Hypothesis H3 and H4 in~\cite{Balachandra} are satisfied. Finally, H5 is also satisfied because $g(t_f,\tau,y,z_f,0)$ is independent of $y$. Under Hypotheses H1 through H5, Theorem 1 in~\cite{Balachandra} provides the expected result.

\vspace{-0.1cm}
\section{Numerical illustrations}\label{sec:numerical-illustration}
\vspace{-0.1cm}
Consider the multi-agents system~\eqref{eq:def_sys_derivative} with $8$ agents and two clusters whose communication pattern is defined in Figure~\ref{fig:interaction-graph}. When not null, internal communication weights are set to $a_{ij}(t) = 2+\cos(2t)$ in cluster $1$, to $a_{ij}(t) = 1$ in cluster $2$ and external communication weights are set to $a_{ij}(t) = \eps (sin(t)+2)/3$ where $\eps$ is a constant parameter. As a consequence, Assumption~\ref{ass:intra-reciprocity},~\ref{ass:inter-reciprocity},~\ref{ass:cI-bound-eps-uniform-bound-on-delta} and~\ref{ass:averaging} are all satisfied. In the present case, internal weights vary in time all with the same multiplicative factor so that $q^k(t)$ is time-invariant and correspond to the left eigenvector of $L^k(0)$ associated to eigenvalue $0$ such that $q^k \un_{n_k} = 1$ \sam{(see Remark~\ref{r:computing-qk} for the computation of $q^k$ in the general case)}. We simulate the system with initial conditions $x(0) = (6,6.3,4.4,5.2,3,3.5,0.4,2.2)^T$ (see Figure~\ref{fig:simu}). As expected \sam{from Proposition~\ref{prop:consensus-ys} and Lemma~\ref{l:LE0-implies-doty-0}, all aggregate variables converge to consensus and disagreement variables converge to $0$, even when $\eps$ is not small}. \sam{When $\eps=1$ the approximations are not precise, but they approch the real variables uniformly as $\eps$ diminishes, as expected from Theorem~\ref{th:main}.}

\begin{figure}[!ht]
\begin{center}
\includegraphics[scale=0.4,clip = true, trim=4.5cm 6.4cm 5cm 5.9cm,keepaspectratio]{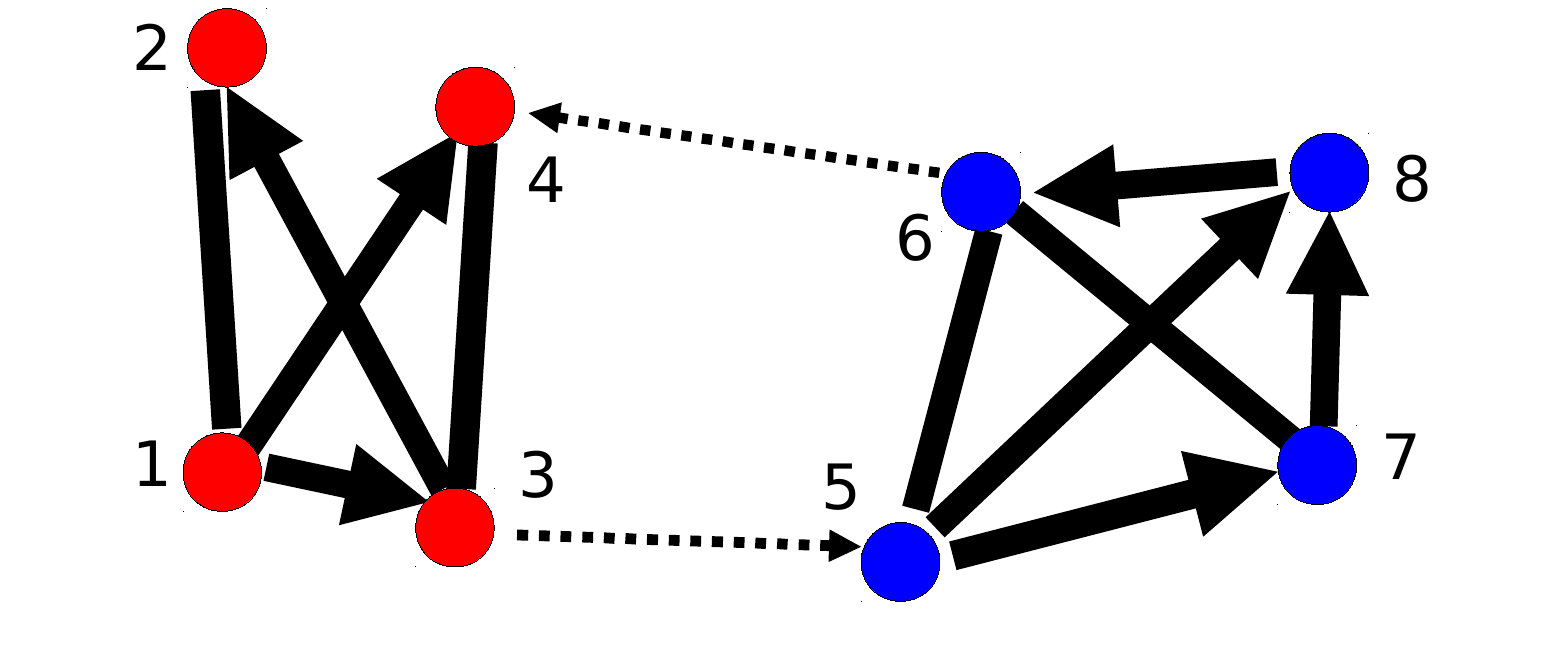}
\vspace{-0.3cm}
\caption{Communication graph. Red and blue circles are agents in cluster $1$ and $2$ respectively. Bold and dashed thin arrows correspond to communication inside and between clusters, respectively. Arrows are replaced by lines when the communication is bidirectional. Only communications with non-zero weight are displayed.}\label{fig:interaction-graph}
\end{center}
\end{figure}

\begin{figure}[!ht]
\begin{center}
\includegraphics[scale=0.3,clip = true, trim=4.5cm 8cm 4cm 8cm,keepaspectratio]{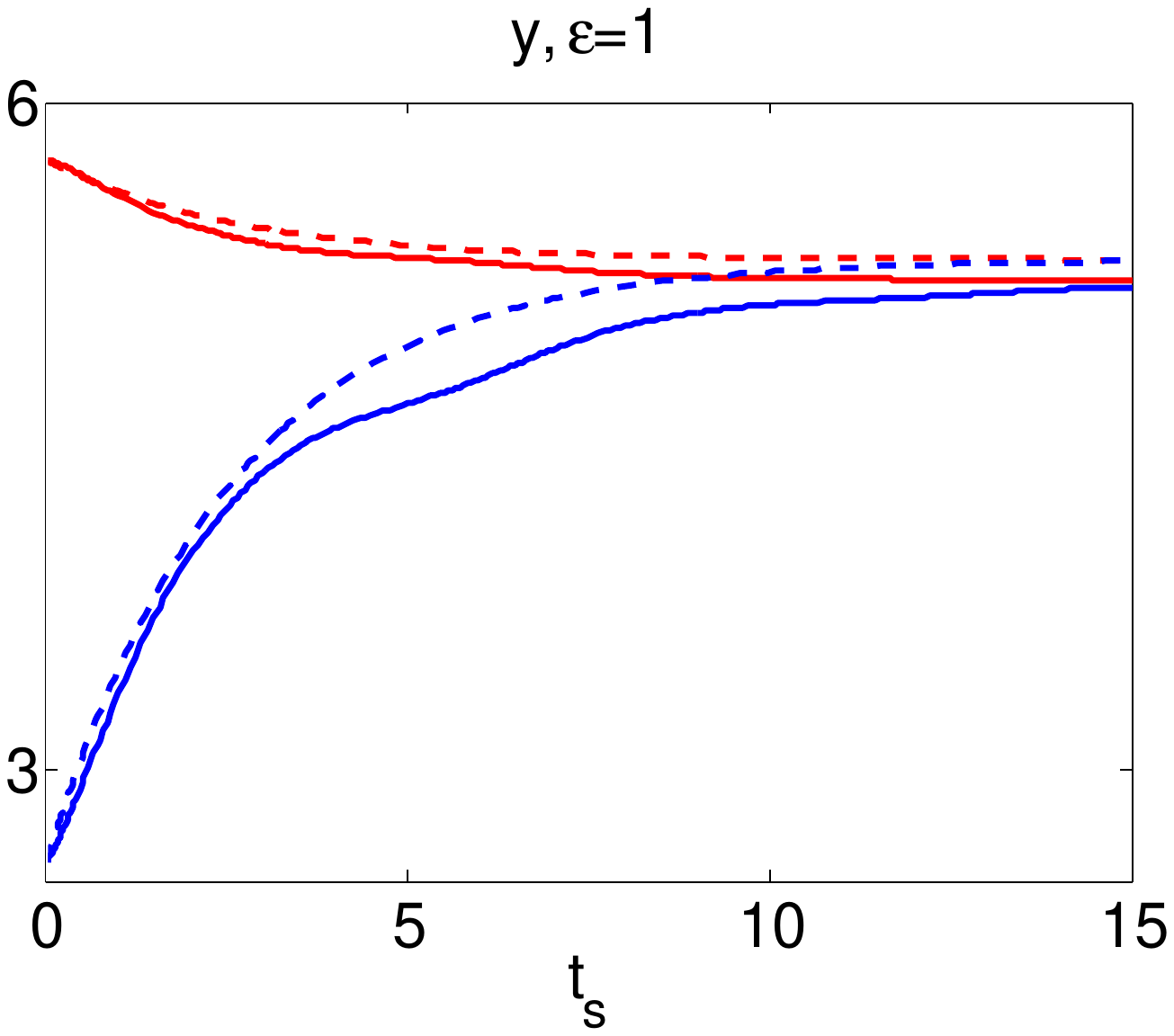}
\includegraphics[scale=0.3,clip = true, trim=4cm 8cm 4cm 8cm,keepaspectratio]{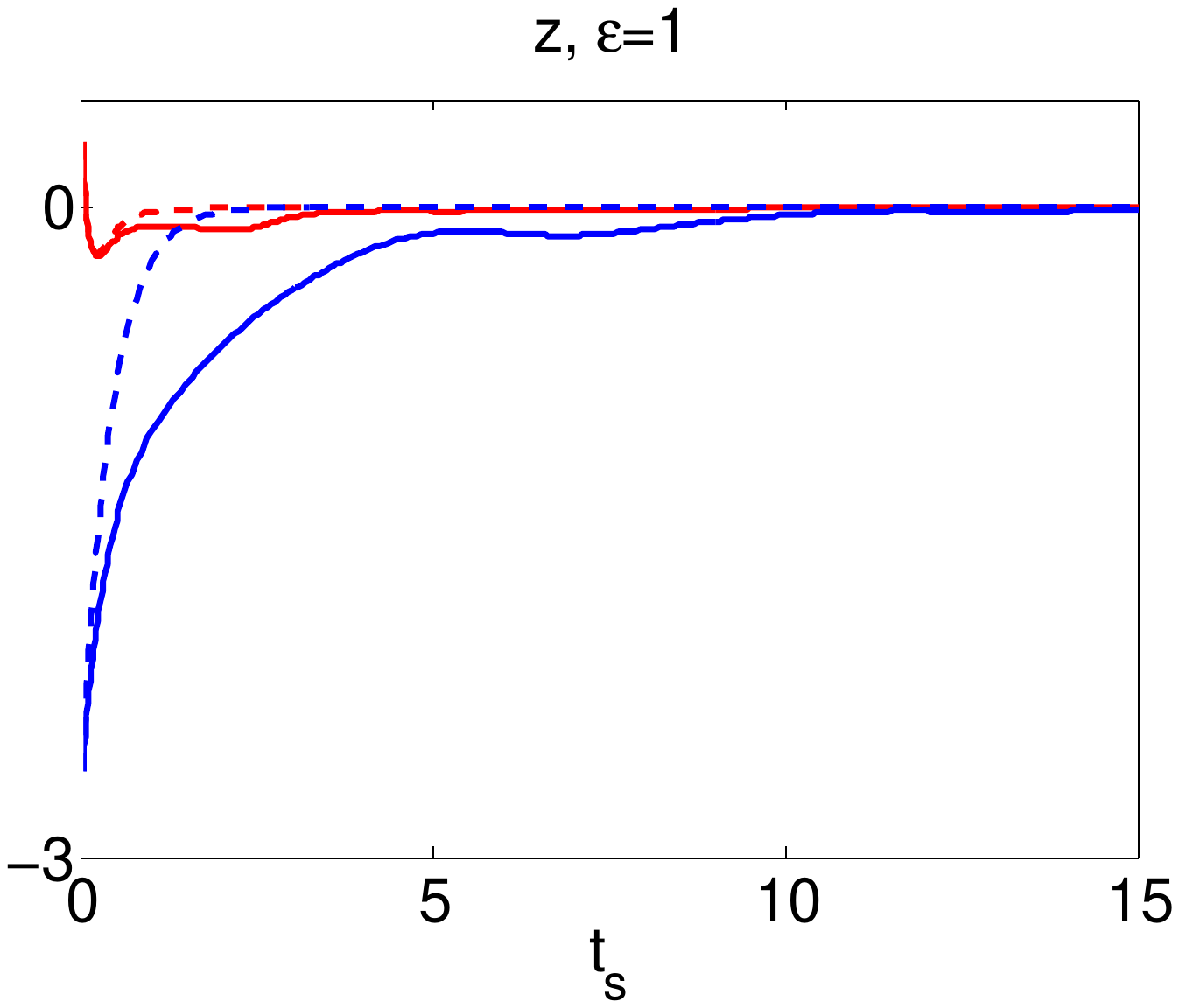}
\includegraphics[scale=0.3,clip = true, trim=4.5cm 8cm 4cm 8cm,keepaspectratio]{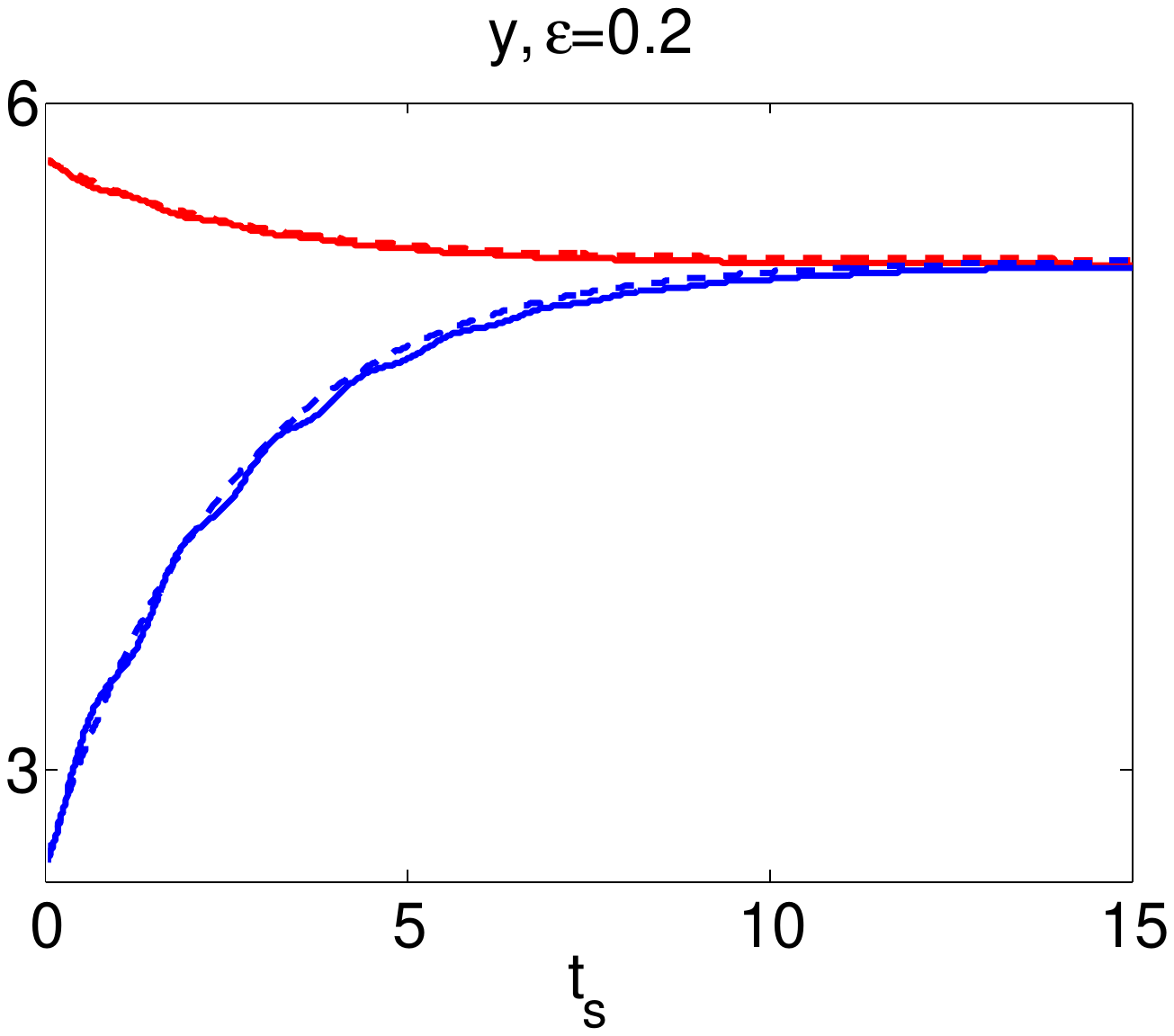}
\includegraphics[scale=0.3,clip = true, trim=4cm 8cm 4cm 8cm,keepaspectratio]{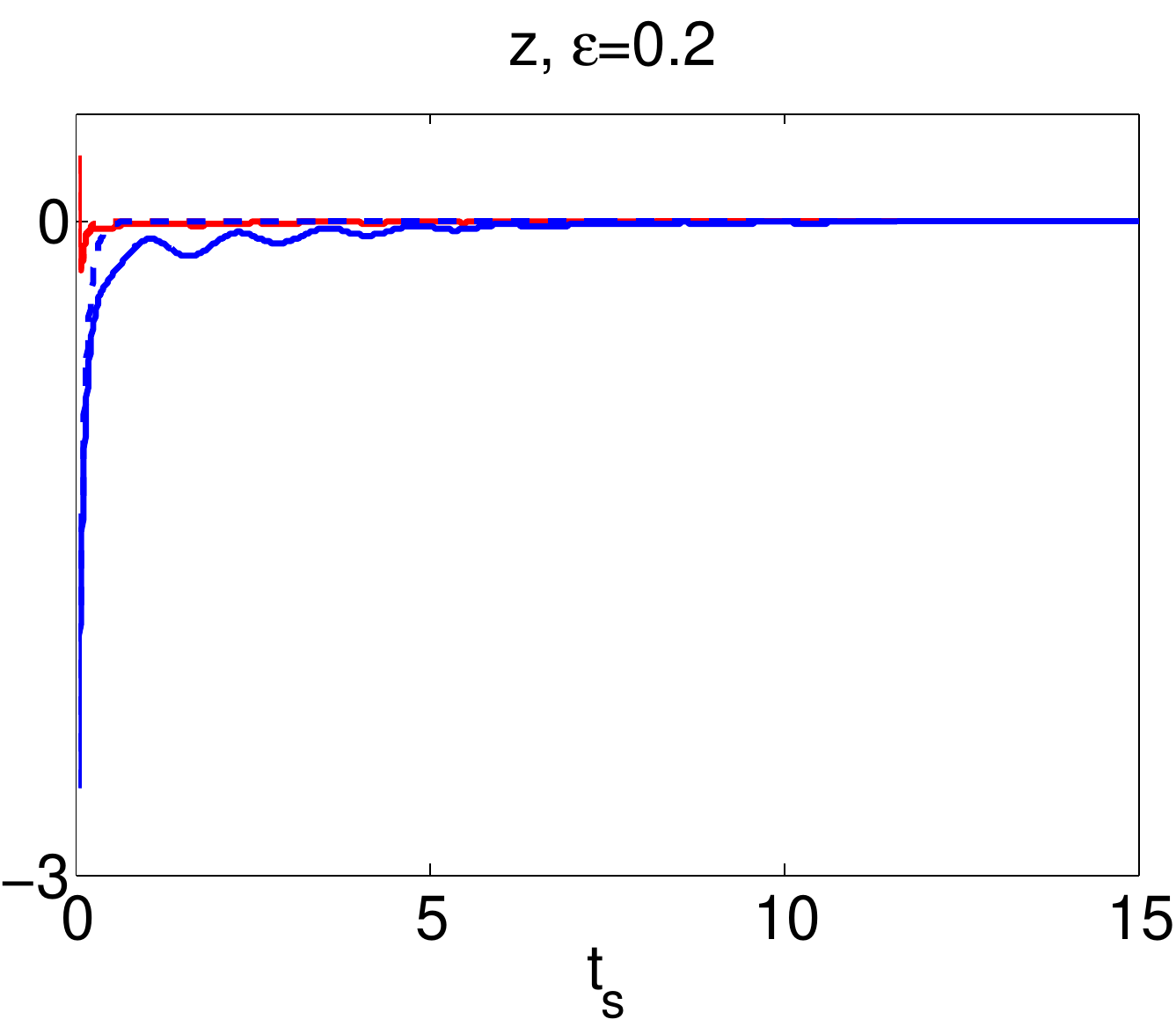}
\includegraphics[scale=0.3,clip = true, trim=4.5cm 8cm 4cm 8cm,keepaspectratio]{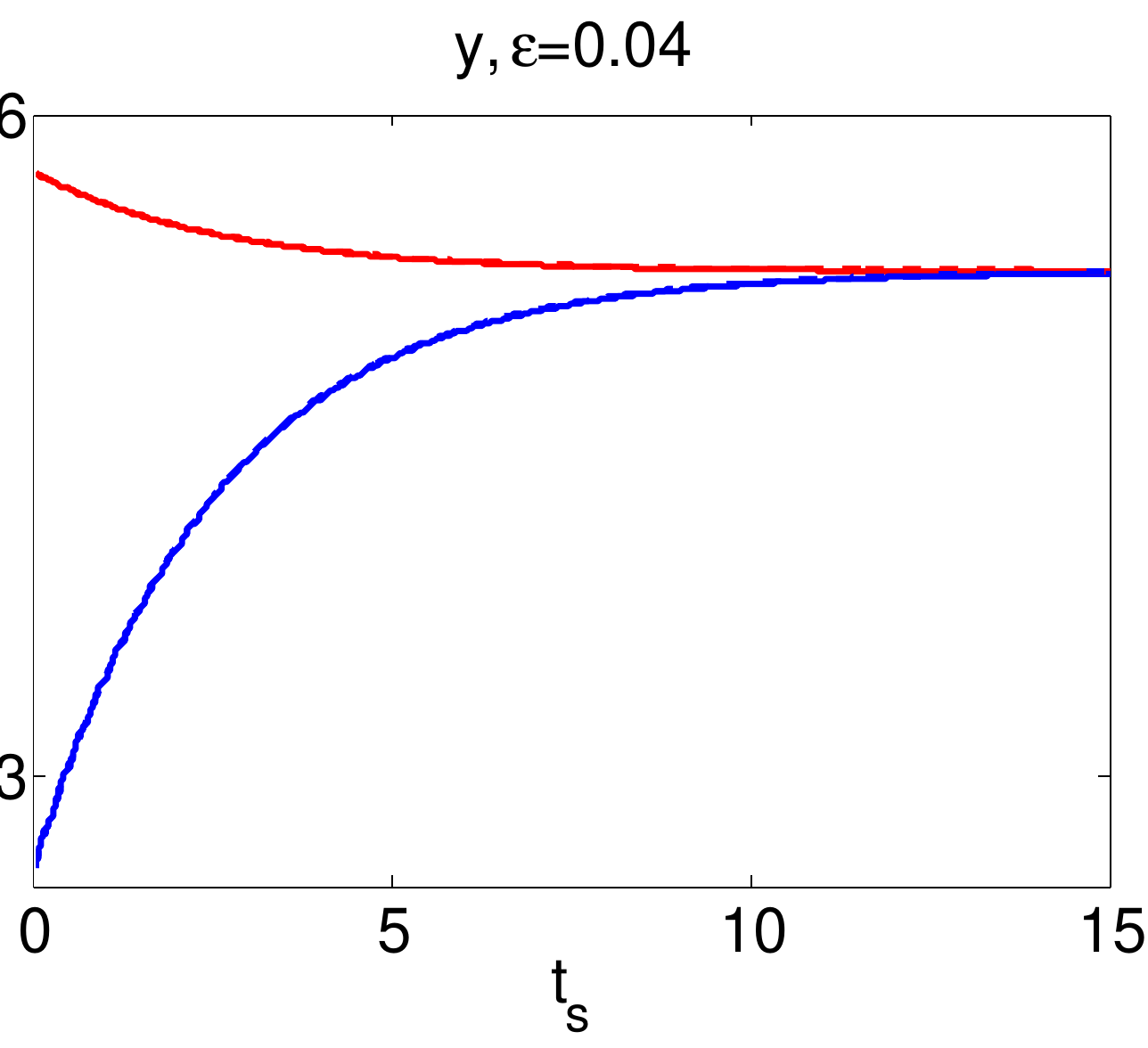}
\includegraphics[scale=0.3,clip = true, trim=4cm 8cm 4cm 8cm,keepaspectratio]{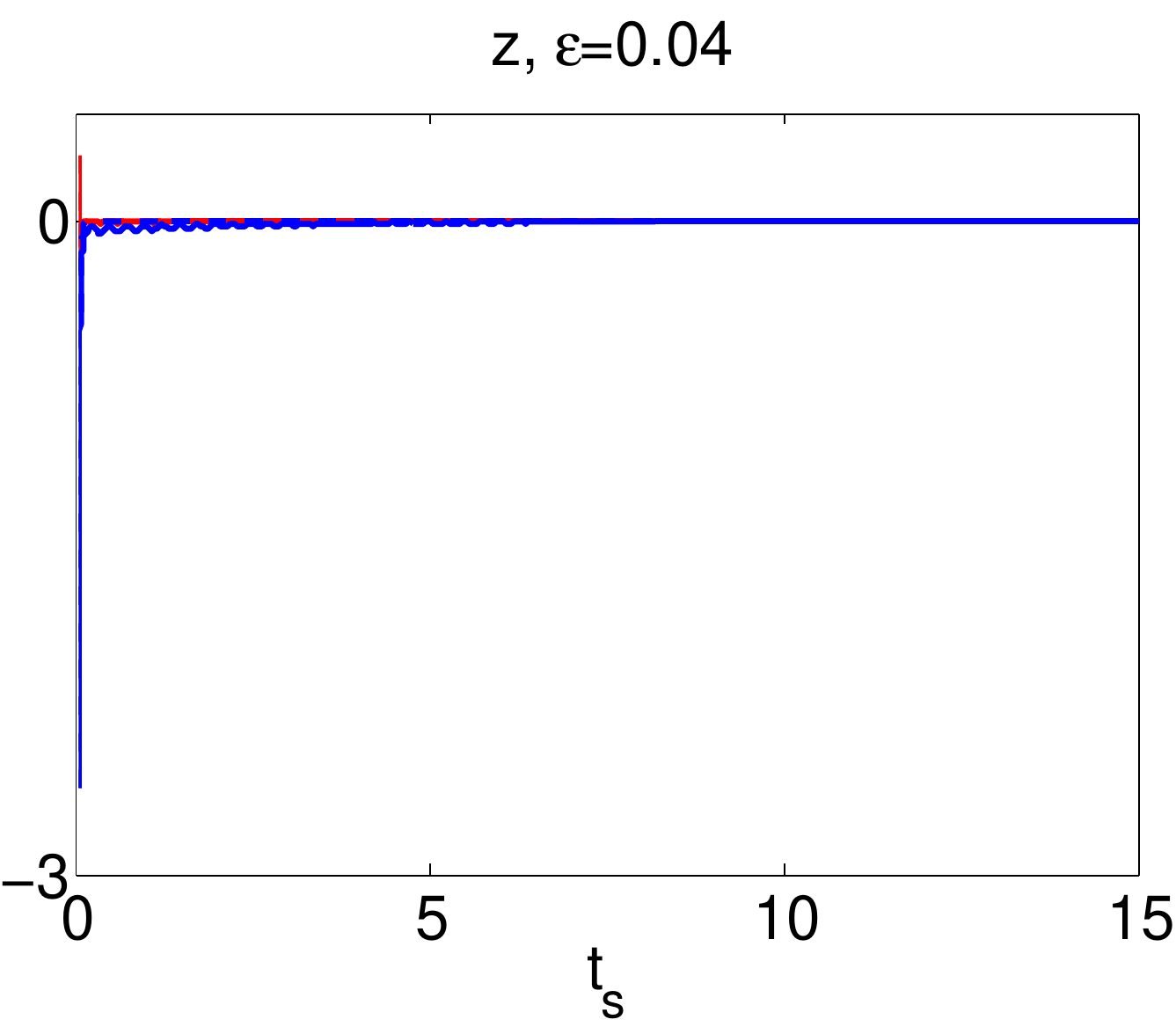}
\vspace{-0.3cm}
\caption{Trajectories in the slow time scale $t_s$ of the aggregate (left) and disagreement variables (right) along with their approximations, for $\eps \in \{1,0.2, 0.04\}$ from top to bottom. Strait lines correspond to the original variables while dashed lines are the approximations. Red lines are for cluster $1$ and blue lines for cluster $2$. Only the first disagreement variable of each cluster is displayed (the two others present a similar behavior).}\label{fig:simu}
\end{center}
\end{figure}
\vspace{-0.2cm}
\section{Conclusion}

In the present study, we have extended the singular perturbation analysis for linear consensus systems found in~\cite{ChowKokotovic} to time-varying and directed networks. We have emphasized that the right tool to study the network at hand is averaging theory for two time-scales systems. 

\appendix

In the Appendix, we provide intermediate results: Lemma~\ref{l:Phi_bound_eta}, Lemma \ref{l:bound-Phi_S_b} and Corollary~\ref{cor:bound_phi_eta} which allow to obtain Lemma~\ref{l:invariant_qk_and_bounds}. We also provide the proof of Lemma~\ref{l:orders}.
Let $t\ge 0$ and $S\sqsubset \C_k$. Denote 
$$\tau_S(t) = \min\left\{ \tau\ge t | \int_{t}^\tau \sum_{i \in \C_k\setminus S, }\sum_{j\in S} a_{ij}(s)ds = 1 \right\}.$$
\begin{lemma}\label{l:Phi_bound_eta}
Let $t\ge 0$. For all $S\sqsubset \C_k$, there exists $i \in \C_k \setminus S$ such that for all $j \in S \cup \{i\}$,
$$
\ssum_{h \in S} \Phi_{jh}(t,\tau_S(t)) \ge \eta,
$$
where $\eta = \exp(-K_I) / n_k$.
\end{lemma}

\begin{proof}
Following ideas inLemma $6$ and Remark $3$ in~\cite{MartinHendrickx2015}, we define an artificial trajectory $u$ satisfying system~\eqref{sys:isolated-internal-dynamics} over $[t,\tau_S(t)]$ with initial states $u_h(t) = 1$ for $h \in S$ and $u_i(t)  = 0$ for $i \in \C_k\setminus S$. Then, for all $s \in [t,\tau_S(t)]$, and $j \in \C_k$, $u_j(s) \in [0,1]$.
According to~\cite[Lemma~6]{MartinHendrickx2015},
$$\ssum_{h \in S} \Phi_{jh}(t,\tau_S(t)) = u_i(\tau_S(t)).$$
It remains to show there exists $i \in \C_k \setminus S$ such that
\begin{equation}\label{eq:bound-ui-eta}
 u_i(\tau_S(t)) \ge \eta.
\end{equation}
First we start by showing that for all $s\in [t,\tau_S(t)]$ and $h \in S$,
\begin{equation}\label{eq:bound-on-uhs}
 u_h(s) \ge e^{-K_I}.
\end{equation}
Denote $\bar{h}(s) \in S$ and $\underbar{u}(s)$ such that
$$
\underbar{u}(s) = u_{\bar{h}}(s) = \min_{h \in S} u_{h}(s).
$$
Recall that (see~\cite{HendrixTsitsiklis_CutBalanced_IEEETAC2013}), for almost all time $s\in [t,\tau_S(t)]$,
$$
\dot{\bar{u}}(s) = \ssum_{i\in \C_k} a_{\bar{h}i}(s) (u_i(s) - \bar{u}(s)).
$$
Since $u_i(s) \ge 0$ and $a_{\bar{h}i}(s) \ge 0$, we have
\begin{eqnarray*}
\dot{\bar{u}}(s) &\ge& - \ssum_{i\in \C_k \setminus S} a_{\bar{h}i}(s) \bar{u}(s)\\
&\ge& - \ssum_{i\in \C_k \setminus S} \ssum_{h \in S} a_{hi}(s) \bar{u}(s).
\end{eqnarray*}
A comparison theorem along with Assumption~\ref{ass:intra-reciprocity} and the definition of $\tau_S(t)$ allow to deduce
$$
\bar{u}(s) \ge e^{-K_I} \bar{u}(t) = e^{-K_I},
$$
by definition of $\bar{u}$ and the choice of initial condition for $u$. Equation~\eqref{eq:bound-on-uhs} is now granted.

To show a contradiction and prove equation~\eqref{eq:bound-ui-eta}, assume that for all $i \in \C_k \setminus S$ and $s\in [t,\tau_S(t)]$, 
\begin{equation}\label{eq:bound-ui-for-contradiction}
u_i(s) \le \frac{e^{-K_I}}{n_k} < e^{-K_I}.
\end{equation}
Denote $\Sigma(s) = \sum_{i\in \C_k \setminus S} {K_I}^{-i} u_i(s)$.
Under the previous condition and equation~\eqref{eq:bound-on-uhs}, we can apply~\cite[Lemma~$8$]{SamAntoine_Persistent_SICON2013} so that the conclusion of~\cite[Lemma $9$]{SamAntoine_Persistent_SICON2013} rewrites as
\begin{eqnarray*}
\dot{\Sigma}(s)
&\ge& 
\left(\underbar{u}(s) - \max_{i\in \C_k \setminus S} u_i(s) \right) 
\ssum_{i\in \C_k \setminus S} \ssum_{h \in S} a_{ih}(s)\\
&\ge& 
\frac{(n_k-1) e^{-K_I}}{n_k}\ssum_{i\in \C_k \setminus S} \ssum_{h \in S} a_{ih}(s).
\end{eqnarray*}
Denote $j \in \argmax_{i\in \C_k \setminus S} \{K_I^{-i} u_i(\tau_S(t))\}$. 
Since $\Sigma(t) = 0$, by integration, we obtain
\begin{eqnarray*}
K_I^{-j} u_j(\tau_S(t)) &\ge& \frac{e^{-K_I}}{n_k}\int_t^{\tau_S(t)} \ssum_{i\in \C_k \setminus S} \ssum_{h \in S} a_{ih}(s) ds.
\end{eqnarray*}
Using the definition of $\tau_S(t)$, and $K_I \ge 1$, we obtain
$$
u_j(\tau_S(t)) \ge  \frac{e^{-K_I}}{n_k}.
$$
So that, either the previous equation is correct or equation~\eqref{eq:bound-ui-for-contradiction} is wrong, in any case, equation~\eqref{eq:bound-ui-eta} is satisfied and the lemma is proven.
\end{proof}



\begin{lemma}\label{l:bound-Phi_S_b}
Let $p,d \in \{1, \ldots, n_k\}$ such that $p<d$. Let $S_p \sqsubset \C_k$ with $|S_p| = p$ and $t_p \ge 0$. Then, there exists a growing sequence of sets
$S_{p+1} \ldots, S_d \sqsubset \C_k$ of cardinal $|S_b| = b$ such that for all $b \in \{p, \ldots d-1\}$, and $S_b \subseteq S_{b+1}$, which verifies that for all $b \in \{p+1, \ldots d\}$, 
\begin{equation}\label{eq:label-bound-Phi_S_b}
\forall j \in S_b,
\sum_{h \in S_p} \Phi_{jh}(t_p, t_b) \ge \eta^{b-p},
\end{equation}
where $\eta = e^{-K_I}/n_k$ and $t_b=(\tau_{S_{b-1}}\circ \ldots\circ\tau_{S_p})(t_p)$
with $\circ$ standing for the composition of functions.
\end{lemma}

\begin{proof}
We prove the lemma by induction on $b$. For $b=p+1$, the equation~\eqref{eq:label-bound-Phi_S_b} is obtained as a direct consequence of Lemma~\ref{l:Phi_bound_eta} with $S := S_p$ and $S_{p+1} = S \cup \{i\}$. Assume that equation~\eqref{eq:label-bound-Phi_S_b} is true for some $b \in \{p+1, \ldots d-1\}$. We apply Lemma~\ref{l:Phi_bound_eta} with $S := S_{b}$ and $t:= {t}_b$ to obtain the existence of an element $i' \in \C_k \setminus S_{b}$ such that for all $j \in S_b \cup \{i '\}$,
\begin{equation}\label{eq:sum_phi_eta}
\sum_{h \in S_b} \Phi_{i'h}(t_b, t_{b+1}) \ge \eta.
\end{equation}
Denote $S_{b+1} = S_b \cup \{i '\}$.
Let $j \in S_{b+1}$. To prove equation~\eqref{eq:label-bound-Phi_S_b} for $b := b+1$, notice that by definition of the fundamental matrix,
$$
\sum_{h \in S_p} \Phi_{jh}(t_p, t_{b+1}) = 
\sum_{h \in S_p} \sum_{l \in S_b} \Phi_{jl}(t_b, t_{b+1}) \Phi_{lh}(t_p, t_b).
$$
Applying inequality~\eqref{eq:sum_phi_eta} and then equation~\eqref{eq:label-bound-Phi_S_b}, we obtain
\begin{eqnarray*}
\sum_{h \in S_p} \Phi_{jh}(t_p, t_{b+1}) &=& \eta \sum_{h \in S_p} \Phi_{lh}(t_p, t_b) \\
&\ge& \eta \cdot \eta^{b-p} = \eta^{(b+1)-p}.
\end{eqnarray*}
\end{proof}

\begin{corollary}\label{cor:bound_phi_eta}
Let $t\ge0$. 
There exists a finite time $t'\ge t$ such that 
for all $r,j \in \C_k$,
\begin{equation*}
\forall s\ge t' 
, \Phi_{jr}(t,s) \ge \left(\exp(-K_I) / n_k\right)^{n_k-1}.
\end{equation*}
\end{corollary}
\begin{proof}
Let $r \in \C_k$. Applying Lemma~\ref{l:bound-Phi_S_b} with $p=1$, $d = n_k$ provides the existence of $\tau \ge t$ such that for all $h \in \C_k$,
$$
\Phi_{hr}(t,\tau) \ge \eta^{n_k-1}.
$$
Let $s\ge \tau$ and $h \in \C_k$. The fundamental matrix definition provides
\begin{eqnarray*}
\Phi_{jr}(t,s) &=& \ssum_{h\in \C_k}  \Phi_{jh}(\tau,s) \Phi_{hr}(t,\tau)\\
&\ge& \eta^{n_k-1}\ssum_{h\in \C_k}  \Phi_{jh}(\tau,s) = \eta^{n_k-1}.
\end{eqnarray*}
Taking the largest $\tau$ for all $r \in \C_k$ allows to conclude.
\end{proof}

{\it Proof of Lemma \ref{l:invariant_qk_and_bounds}:} 
The property of the fundamental matrix in equation~\eqref{eq:Phi-sum-to-one} along with equation~\eqref{eq:def_of_q_k} show that
$$
q^k(t)^T \cdot \un_{n_k} = \sum_{j \in \C_k} \Phi^k_{ij}(t,\infty) = 1.
$$
We now derive the lower bound $q_{min}$. Since for all $i,j \in \C_k$, $\Phi^k_{ij}(t,\infty) = q_j^k(t)$, Corollary~\ref{cor:bound_phi_eta} when $s\rightarrow \infty$ directly provides
$$
q_j^k(t) \ge \left(\exp(-K_I) / n_k\right)^{n_k-1} \ge q_{min}.
$$

We now turn to the upper bound $q_{max}$. Let $k$ in $\MM $ and $i \in \C_k$. We have
\begin{eqnarray*}
q^k_i(t) &=& 1 - \sum_{j \in \C_k \setminus \{i\}} q^k_j(t)\\
&\le& 1 - (n_k - 1) q_{min} \\
&\le& 1 - (\min_{k \in \MM } n_k - 1) q_{min}.
\end{eqnarray*}
Finally, we show that $q^k(t)^T \cdot \tilde{x}_{\C_k}(t)$ is invariant in time.
For any trajectory $f$, denote $f(\infty) = \lim_{t \rightarrow \infty} f(t)$.
By definition of the fundamental matrix (equation~\eqref{eq:fundamental-matrix}), and of $q^k(t)$ (equation~\eqref{eq:def_of_q_k}),
\begin{eqnarray*}
q^k(\infty)^T \cdot \tilde{x}_{\C_k}(\infty) &=& q^k(\infty)^T \cdot \Phi^k(t,\infty) \tilde{x}_{\C_k}(t) \\
&=& q^k(\infty)^T \cdot \un_{n_k} \cdot q^k(t)^T \tilde{x}_{\C_k}(t) \\
&=& q^k(t)^T \cdot\tilde{x}_{\C_k}(t).
\end{eqnarray*}

{\it Proof of Lemma~\ref{l:orders}:} 
Notice that $L(t)H=L^E(t)H$ and $J(t)L(t)=J(t)L^E(t)$. Thus,
\[
\begin{split}
\|\bar{A}_{11}(t)\|_\infty&=\|J(t)L^E(t)H\|_\infty=\max_{i\in \MM}\sum_{j=1}^m |\left(J(t)L^E(t)H\right)_{ij}|\\
&\le\max_{i\in \MM}\sum_{j=1}^m\sum_{k=1}^n \sum_{h=1}^n| J_{ik}(t)| |L_{kh}^E(t)| |H_{hj}|\\
&\leq\max_{i\in \MM}\left(\sum_{k\in\C_i}\sum_{h\notin\C_i}  |L_{kh}^E(t) | +\sum_{k\in\C_i}\sum_{h\in\C_i}  |L_{kh}^E(t) | \right)\\
&=\max_{i\in \MM}\left(\sum_{k\in\C_i}\sum_{h\notin\C_i}  |L_{kh}^E(t) | +\sum_{k\in\C_i}  |L_{kk}^E(t) | \right)\\
&=2\max_{i\in \MM}\sum_{k\in\C_i}\sum_{h\notin\C_i}  |L_{kh}^E(t) | \leq2\gamma^E(t)\leq 2c^I(t)\eps.
\end{split}\]
Similarly, using that for all $i\in\{1,\ldots,n-m\},\ j\in\{1,\ldots,n\}$ the components of $Q$ satisfy $|Q_{ij}|\leq 1$, one obtains:
\[
\begin{split}
\|\bar{A}_{21}(t)\|_\infty&=\|QL^E(t)H\|_\infty=\hspace{-0.3cm}\mmax_{i\in\{1,\ldots,n-m\}}\sum_{j=1}^m |\left(QL^E(t)H\right)_{ij}|\\
&\hspace{-1cm}\le\mmax_{i\in\{1,\ldots,n-m\}}\sum_{k=1}^n \sum_{h=1}^n| Q_{ik}| |L_{kh}^E(t)| |H_{hj}|\\
&\hspace{-1cm}\leq2\mmax_{i\in\{1,\ldots,n-m\}}\sum_{k\in\C_i}\sum_{h\notin\C_i}  |L_{kh}^E(t) |\leq2\gamma^E(t)\leq 2c^I(t)\eps.
\end{split}
\]
Using a similar reasoning and the structure of $J(t)$ and $Q$ one also obtains
$
\|\bar{A}_{12}(t)\|_\infty\leq 2c^I(t)\eps.
$

Finally, using  $L(t)=L^I(t)+L^E( t)$ one has
\[\begin{split}
\|\bar{A}_{22}(t)\|_\infty&=\|QL(t)\tilde{Q}(t)\|_\infty\\ &\geq\|QL^I(t)\tilde{Q}(t)\|_\infty-\|QL^E(t)\tilde{Q}(t)\|_\infty.
\end{split}\]
Straightforwardly one gets that 
\[\begin{split}
\|QL^E(t)\tilde{Q}(t)\|_\infty&\leq \|Q\|_\infty\|L^E(t)\|_\infty\|\tilde{Q}(t)\|_\infty\\
&\leq 2\cdot 2\gamma^E(t)\cdot 2\leq 8 c^I(t)\eps.
\end{split}\]
On the other hand $QL^I(t)\tilde{Q}(t)$ is block diagonal being equal to $diag(Q_1L^1(t)\tilde{Q}_1(t),\ldots,Q_mL^m(t)\tilde{Q}_m(t))$. Therefore,  \[\|QL^I(t)\tilde{Q}(t)\|_\infty=\max_{k\in\MM}\|Q_kL^k(t)\tilde{Q}_k(t)\|_\infty.\]
Direct computation shows that $\forall\ i,j\in\{1,\ldots,n_k-1\}$,
\[
(Q_kL^k(t)\tilde{Q}_k(t))_{ij}=L^k_{i+1\ j+1}-L^k_{1\ i+1}.
\]
Consequently,
\[\begin{split}
\|Q_kL^k(t)\tilde{Q}_k(t)\|_\infty&=\max_{i=1}^{n_k-1}\sum_{j=1}^{n_k-1}|L^k_{i+1\ j+1}-L^k_{1\ j+1}|\\
&\hspace{-2cm}\geq \max_{i=1}^{n_k-1}\left|\sum_{j=1}^{n_k-1}L^k_{i+1\ j+1}-L^k_{1\ j+1}\right|\\
&\hspace{-2cm}= \max_{i=1}^{n_k-1}\left|-L^k_{i+1\ 1}-\sum_{j=1}^{n_k-1}L^k_{1\ j+1}\right|\geq c^I,
\end{split}\]
yielding $\|QL^I(t)\tilde{Q}(t)\|_\infty\geq c^I$.

\bibliographystyle{IEEEtran}
\bibliography{consensus_singular_perturbation_cdc_long_for_arxiv}

\end{document}